\documentclass[a4paper,12pt]{article}

\usepackage{amsmath}
\usepackage{times}		
\usepackage{parskip}		
\usepackage[pdftex]{graphicx}	
\usepackage[pdftex]{hyperref}	
\usepackage{amsfonts}
\usepackage[a4paper]{geometry}
\usepackage{amsmath}

\newtheorem{theorem}{Theorem}[section]
\newtheorem{lemma}[theorem]{Lemma}

\newenvironment{proof}[1][Proof]{\begin{trivlist}
\item[\hskip \labelsep {\bfseries #1}]}{\end{trivlist}}
\newenvironment{definition}[1][Definition]{\begin{trivlist}
\item[\hskip \labelsep {\bfseries #1}]}{\end{trivlist}}

\newcommand{\qed}{\nobreak \ifvmode \relax \else
      \ifdim\lastskip<1.5em \hskip-\lastskip
      \hskip1.5em plus0em minus0.5em \fi \nobreak
      \vrule height0.75em width0.5em depth0.25em\fi}

\begin{document}

\author{D. M. Walsh\thanks{E-mail
: darraghmw@gmail.com}\\ \small Economics
Department,\\ \small Trinity College Dublin,
Dublin 2, Ireland}

\title{On the stability of solutions of the Lichnerowicz-York equation}
\maketitle

\maketitle
\begin{abstract}
We study the stability of solution branches for the Lichnerowicz-York equation at moment of time symmetry with constant unscaled energy density. We prove that the weak-field lower branch of solutions is stable whilst the upper branch of strong-field solutions is unstable. The existence of unstable solutions is interesting since a theorem by Sattinger proves that the sub-super solution monotone iteration method only gives stable solutions. 

\end{abstract}

\section{Introduction}
The loss of uniqueness is a  salient feature of many nonlinear systems and is often accompanied by a loss or ``exchange'' of stability as two solution branches meet. A number of recent papers have looked at the uniqueness of solutions of conformal formulations of the Einstein constraints.

Recent numerical work by Holst and Kungurtsev in \cite{Holst} has confirmed the assumption made in \cite{W} that the linearisation of the (constant) energy density moment of time symmetry Lichnerowicz-York equation has a one dimensional kernel and that a quadratic fold results giving two solutions for subcritical energy density $\rho<\rho_c$, one solution at the critical value $\rho=\rho_c$  and none for energy density greater than a  critical value $\rho>\rho_c$.

In \cite{Maxwell}, Maxwell studied the conformal formulations of the Einstein constraint equations without the usual constant mean curvature condition (CMC) so that the constraints do not decouple. One motivation for that study was to examine the uniqueness of solutions in the far-from-constant mean curvature regime now that existence results for this case exist (see \cite{HNT}, \cite{Maxwell2}). He found interesting non-existence and non-uniqueness results showing that the constraints are ill-posed beyond the non-uniqueness that is introduced when one couples the lapse fixing equation to the four constraint equations, as in the extended conformal thin sandwich (XCTS) formulation (see \cite{PY}, \cite{BOMP}, \cite{W}) and some constrained evolution schemes (see \cite{Rinne}, \cite{meudon} for a resolution of this scaling problem).

The inherent ill-posedness of conformal formulations of the constraints found in \cite{Maxwell} is worthy of further analysis. In this work we continue our bifurcation analysis of the Einstein constraints, begun in \cite{W} with an analysis of the XCTS system, by studying a familiar problem and introducing the important related issue of the stability of the solutions obtained. The ultimate goal of this program is to better understand the mapping between general free initial data and solutions of the constraint equations in their various conformal formulations. See also the recent work of Holst and Meier \cite{HM} in this regard.

In section 2 we study the solution branches of the constant density star found in \cite{BOMP}, \cite{W} and \cite{Holst}. We prove that the lower branch of weak-field solutions is stable whereas the upper branch is unstable. (We also prove that the kernel of the linearisation is one-dimensional, which was shown numerically in \cite{Holst}). In section 3 we follow \cite{Kielhofer} in applying Liapunov-Schmidt methods, as in \cite{W}, to the stability analysis and find that the exchange of stability that occurs at the critical point of a  fold is in fact generic under mild non-degeneracy conditions. The existence of unstable solutions to the constraint equations is interesting in itself because the most popular method for proving existence, the sub-super solution monotone convergence method, only yields stable solutions (a result proven by Sattinger in \cite{Sattinger1}).

\section{Stability of solutions}


%

%




We shall be concerned with the Hamiltonian constraint, the Lichnerowicz-York equation, at moment of time symmetry, for a conformally-flat background metric and an unscaled constant energy density $\rho$:
\begin{equation}\label{nonlinear}
F(\phi,\rho):=\nabla^2\phi+2\pi\rho\phi^5=0
\end{equation}
 with $\phi>0$.
 We work in spherical symmetry in this section so that $\nabla^2=\frac{d^2}{d r^2}+\frac{2}{r}\frac{d}{dr}$. For simplicity we assume that $\rho$ is zero outside a  ball of radius $r =  1$ with boundary conditions $\phi(1)=1$, and $\frac{d\phi}{d r}(0)=0$ as in \cite{Holst}. 

We are interested in the stability of stationary solutions of the following parabolic problem:
\begin{equation}\label{para}
\frac{\partial\nu}{\partial t}=\nabla^2\nu+2\pi\rho\nu^5,
\end{equation}
with $\nu(t=0,x)=\phi_0(x)$.

We are interested in whether, for initial data $\phi_0(x)$ close to the stationary solution, the solution to this parabolic problem tends to the stationary solution as $t\rightarrow \infty$. To motivate the definition of stability that follows we take
$$\nu=\hat{\nu} + \epsilon \theta_1$$ where $\hat{\nu}$ is the stationary solution and $\theta_1$ is the eigenfunction corresponding to the principal eigenvalue $\mu_1$ (the smallest eigenvalue) and $\epsilon$ is a small constant. Then substituting this into (\ref{para}) gives
$$\frac{\partial\theta_1}{\partial t}=\nabla^2\theta_1+10\pi\rho\phi^4\theta_1=-\mu_1\theta_1$$ to first order in $\epsilon$.

\begin{definition}
 We will say that a stationary solution to (\ref{para}) is stable if the principal eigenvalue of the linearisation 

\begin{equation}\label{lin}
\nabla^2\theta_k+10\pi\rho\phi^4\theta_k=-\mu_k\theta_k,
\end{equation}
with $\theta_k(1)=0$ and $\theta_k'(0)=0$ satisfies $\mu_1(\nabla^2+10\pi\rho\phi^4)>0$ and otherwise that it is unstable.
\end{definition}

In the next Lemma we collect some facts about the eigenfunctions and eigenvalues of (\ref{lin}).

\begin{lemma}$\newline$
\begin{enumerate}
\item{All the eigenfunctions of (\ref{lin}) are orthogonal and all eigenfunctions except $\theta_1$ have nodes. Furthermore, the principal eigenvalue is simple i.e. if f is any solution of 
\begin{equation}
\nabla^2f+10\pi\rho\phi^4f=-\mu_1f,
\end{equation}
then f is proportional to $\theta_1$}
\item{The principal eigenvalue has a variational characterisation given by the Rayleigh quotient:}
\begin{equation}\label{ray}
\mu_1(\nabla^2+a(x))=min_{\eta\neq0\in D}\frac{\int|\nabla\eta|^2-a(x)\eta^2dv}{\int\eta^2dv}
\end{equation}
where D is the set of smooth functions satisfying the boundary conditions $\eta(1)=0,$ $\frac{d\eta}{d r}(0)=0$

\item{If $a(x)\geq b(x)$ with $a(x)>b(x)$ in a subset of positive measure then
$$\mu_1(\nabla^2+a(x))<\mu_1(\nabla^2+b(x))$$} 

\item{ Similarly, if $a(x)\leq b(x)$ with $a(x)<b(x)$ in a subset of positive measure then
$$\mu_1(\nabla^2+a(x))>\mu_1(\nabla^2+b(x)).$$}

\end{enumerate}
\end{lemma}

\begin{proof}

1.  
It is easy to show that the eigenfunctions of (\ref{lin}) are orthogonal:
\begin{align*}
(\mu_k-\mu_j)\int_V \theta_k\theta_j dV & = \int_V \left(\theta_k( \nabla^2\theta_j+10\pi\rho\phi^4\theta_j)-\theta_j(\nabla^2\theta_k+10\pi\rho\phi^4\theta_k )\right)dV \\
&= \int _{\partial V}
(\theta_k\nabla \theta_j-\theta_j\nabla \theta_k).ndS\\
&=0
\end{align*}
 and we have used Green's second identity with $n$ the outward pointing normal to the surface element $dS$ at $r=1$. We know that the eigenfunction corresponding to the first eigenvalue $\mu_1$ is of definite sign so may be taken to be positive (see\cite{Evans}). The orthogonality relation above then says that all higher eigenfunctions must have nodes. The proof that $\mu_1$ is simple may be found in \cite{Evans} for example.

2.  This well known result may be found for example in (\cite{Evans}).

3. and 4. follow directly from the above result.

QED
\end{proof}

In \cite{BOMP} the authors worked in an unbounded domain which complicates the analysis of the spectrum of the linearisation. We work in a  bounded domain corresponding to a   ball of radius $r=1$ with $\phi(1)=1$. It is straightforward, using the ideas in \cite{BOMP}, to generate a  one parameter family of solutions with this boundary condition:

\begin{equation}
\phi(r;\alpha)=\left(\frac{3}{2\pi\rho(\alpha)}\right)^{\frac{1}{4}}\left(\frac{\alpha}{r^2+\alpha^2}\right)^{\frac{1}{2}} 
\end{equation}
where the boundary condition gives
\begin{equation}
\rho(\alpha)=\frac{3}{2\pi}\left(\frac{\alpha}{1+\alpha^2}\right)^2
\end{equation}

Note that the maximum value of $\rho:=\rho_c$ occurs at $\alpha_c=1$ and that the limit $\alpha\rightarrow \infty$ corresponds to a lower branch of solutions and the limit $\alpha\rightarrow 0$ corresponds to an upper branch of solutions.

It is easy to check that 
\begin{equation}
\theta_k(r)=\frac{1-r^2}{(1+r^2)^{\frac{3}{2}}}
\end{equation}

satisfies (\ref{lin}) with $\alpha=1$ in $\rho$ and $\mu_k=0$. This eigenfunction has no nodes and  therefore corresponds to the principle eigenvalue found numerically in \cite{Holst}. So we have that $\mu_1=0$, we have found the principal eigenfunction at the critical value $\rho_c$, $\alpha_c=1$. Since this eigenvalue is simple we have that the kernel of the linearisation at the critical energy density $\rho_c$ is one dimensional.

In \cite{BOMP}, \cite{Holst} and above, global branches of solutions were found corresponding to a quadratic fold. For each $\rho<\rho_c$ we have an upper solution $\phi_U$ and a  lower solution $\phi_L$. (As in \cite{BOMP}, a simple calculation reveals that there are apparent horizon's on the upper branch of solutions at the location $r=\alpha$, and none on the lower branch). 

Before showing that the upper solutions $\phi_U$ are unstable and the lower solutions $\phi_L$ are stable, we motivate our analysis by the following observation.

If we multiply the linearisation (\ref{lin}) by $u:=\phi-1$ and integrate twice by parts we get the following identity, having used (\ref{nonlinear}):
$$2\pi\int\left(\rho\phi^4\theta_1(4u-1)\right)dv=-\mu_1\int\theta_1udv.$$

With $\theta_1$ the first eigenfunction and therefore positive and $u\geq 0$ since $\phi \geq 1$ we see that for weak-field solutions
$(4u-1)$ is likely less than zero so that $\mu_1>0$ and the solution is stable. But we can imagine a strong field solution whereby $(4u-1)>0$ so that $\mu_1<0$ and we have an unstable solution.

\begin{theorem}$\newline$
1. The lower branch of solutions, i.e. $\phi(r;\alpha)$ with $1<\alpha<\infty$, is stable.

2. The upper branch of solutions, i.e. $\phi(r;\alpha)$ with $0<\alpha<1$, is unstable.

\end{theorem}

\begin{proof}
1. We know that $\mu_1(\nabla^2+10\pi\rho_c\phi_c^4)=0$. Motivated by this we examine the difference in the ``potentials'' in the linearisation (\ref{lin})

\begin{align}\label{diff}
\rho(\alpha)\phi(\alpha)^4-\rho_c\phi_c^4 &=
\frac{3}{2\pi}\left(\frac{\alpha^2}{(r^2+\alpha^2)^2}-\frac{1}{(r^2+1)^2}\right)
=\frac{3}{2\pi}\left(\frac{(r^4-\alpha^2)(\alpha^2-1)}{(r^2+1)^2(r^2+\alpha^2)^2}\right)
\end{align} 

Note that $r\in[0,1]$. On the lower branch of solutions $1<\alpha<\infty$, so that $\rho(\alpha)\phi(\alpha)^4-\rho_c\phi_c^4<0$, so Lemma 2.1.3 implies that

\begin{equation}
\mu_1(\nabla^2+10\pi\rho(\alpha)\phi(\alpha)^4)>\mu_1(\nabla^2+10\pi\rho_c\phi_c^4)=0
\end{equation}

so the lower branch of solutions is stable.

2. We now proceed to show that the upper branch of solutions is unstable. The comparison test above yields an indefinite result so we cannot implement Lemma 2.1.4.

Instead we utilise the fact that
\begin{align}\label{ray}
\mu_1(\nabla^2+10\pi\rho\phi^4)&=min_{\eta\neq0\in D}\frac{\int|\nabla\eta|^2-10\pi\rho\phi^4\eta^2dv}{\int\eta^2dv}\\
&=min_{\eta\neq0\in D}\frac{\int|\nabla\eta|^2-15\left(\frac{\alpha}{\alpha^2+r^2}\right)^2\eta^2dv}{\int\eta^2dv}.
\end{align}

Using a trial function $\eta_T\in D$ then provides an upper bound for the principal eigenvalue corresponding to a particular choice of $0<\alpha<1$, i.e. $\mu_1(\alpha)\leq R(\alpha)$ where $R(\alpha)$ is the Rayleigh quotient above. 

We shall choose 
\begin{equation}\label{eta}
\eta_T=\frac{1-r^2}{(\alpha^2+r^2)^\frac{3}{2}}.
\end{equation}

and use MATLAB to perform the integration exactly. This yields a rather complicated expression given in Appendix A. A graph of $R(\alpha)$ is shown in Figure ~\ref{fig 1} and some explicit values for $R(\alpha)$ are given in Table 1 in Appendix A. We clearly see that for $0<\alpha<1$, the upper branch of solutions, that $\mu_1(\alpha_{upper})\leq R(\alpha_{upper})<0$ so we have that the upper branch of solutions is unstable.

QED
\end{proof}
\begin{figure}
\centerline{\includegraphics[width=12cm,height=8cm]{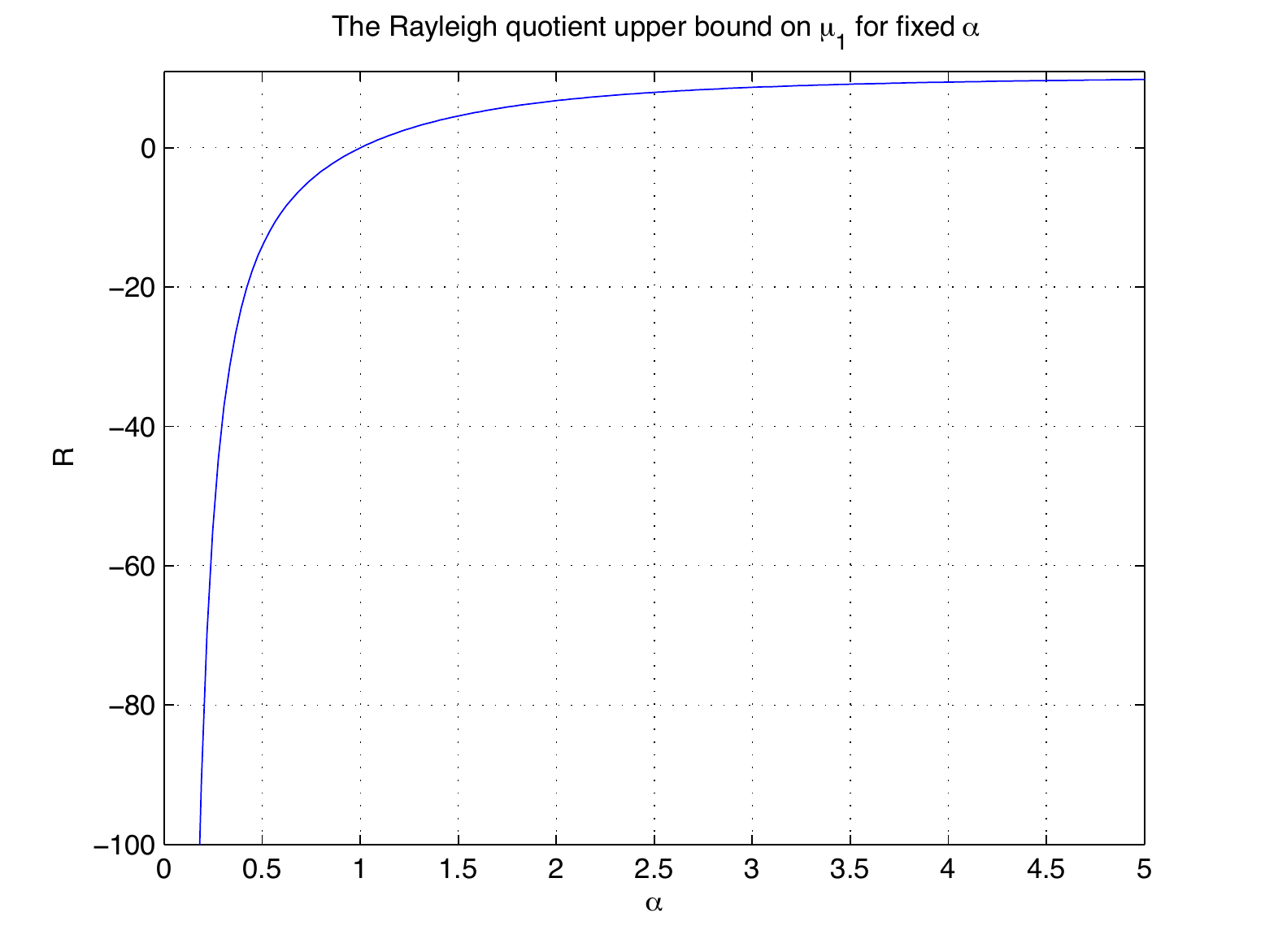}}
\caption{The Rayleigh quotient as a function of $\alpha$ for the trial function $\eta_T$}
\label{fig 1}
\end{figure}



\section{Stability of solutions via Liapunov-Schmidt reduction}

Liapunov-Schmidt (LS) methods were applied to the Lichnerowicz-York equation with an unscaled source in \cite{W} and more recently in \cite{HM}. Here we follow \cite{Kielhofer} and show that locally the exchange of stability witnessed in the previous section is in fact generic for quadratic fold-type branches.

We again consider a nonlinear suitably smooth equation $F(x,\rho)=0$ with $x\in X$, a Banach space and $\rho \in \mathbb{R}$, $(F: X * \mathbb{R}\rightarrow Z)$ with $X\subset Z$
and assume that zero is a  simple eigenvalue of $D_xF(x_0, \rho_0)$. We let $\theta$ denote the one-dimensional kernel and denote by $\theta^*$ the cokernel of $D_{x}F(x_0,\rho_0)$ normalised such that $\int\theta\theta^*=1$ (we assume the linearisation has a Fredholm index of zero). The LS construction allows us to construct a continuously differentiable curve of solutions, for small positive $\delta$, $\{(x(s),\rho(s))| s\in(-\delta,\delta),\ \  (x(0),\rho(0))=(x_c,\rho_c)\}$, such that  
$F(x(s),\rho(s))=0$ for all $s\in(-\delta,\delta)$.

For the LS projection operators and splittings of the range and domain we follow the notation of \cite{Kielhofer} section 1.4-1.7, which is equivalent to \cite{W}. In particular
we have that
\begin{align}
Z&=R(D_xF(x_c, \rho_c))\oplus N(D_xF(x_c, \rho_c))\\
X&=N\oplus(R\cap X),
\end{align}

where N denotes the kernel space and R the range of the linearisation.

Note that if the following nondegeneracy conditions hold:
\begin{equation}\label{nondegen1}
D_{xx}F(x_c,\rho_c)[\theta,\theta] \notin R(D_xF(x_c,\rho_c))
\end{equation}
\begin{equation}\label{nondegen2}
D_\rho F(x_c,\rho_c)\notin R(D_xF(x_c,\rho_c))
\end{equation}

then it is proven in \cite{Kielhofer} that

\begin{equation}
\dot{\rho}(0)=0, \ \ \ \ \ddot{\rho}(0)\neq 0
\end{equation}
where a dot denotes differentiation with respect to s,
and so the tangent vector to the solution curve $(x(s),\rho(s))$ at $(x_c,\rho_c)$ is $(\theta,0)$
(see corollary 1.4.2 in \cite{Kielhofer}) and we have a  turning point or fold (sometimes called a 'saddle-node').

To analyse the stability of the 'branches' of the fold (upper and lower) we note the following result (Proposition 1.7.2 in \cite{Kielhofer}):

\begin{theorem}$\newline$
There exists a  continuously differentiable curve of  perturbed eigenvalues $\{\mu(s)| s\in (-\delta,\delta)\  \ \mu(0)=0\}$ in $\mathbb{R}$ such that
$$D_xF(x(s), \rho(s))(\theta+w(s))=\mu(s)(\theta+w(s)),$$

where $\{w(s)| s\in (-\delta,\delta), \ \ w(0)=0\} \subset R\cap X$ is continuously differentiable (and the size of the interval $(-\delta,\delta)$ is not necessarily the same as in the solution curve above, but possible shrunk). In this sense, $\mu(s)$ is the perturbation of the critical zero eigenvalue of $D_xF(x_c, \rho_c)$.

\end{theorem}
 (Note that we used the opposite sign convection for the eigenvalues of the linearisation in section 2  where $D_{x}F(x_c,\rho_c)\theta=-\mu\theta$).
 
 It is then straightforward to show that (see 1.7.30 in \cite{Kielhofer}),
 \begin{equation}\label{mudot}
 \dot{\mu}(0)=\int D_{xx}F(x_c,\rho_c)[\theta,\theta] \theta^*\neq 0
 \end{equation}
and
\begin{equation}\label{rhoddot}
\dot{\mu}(0)=-\left(\int D_\rho F(x_c,\rho_c)\theta^*\right)\ddot{\rho}(0).
\end{equation}

Now $\mu(0)=0$, $\dot{\mu}(0)\neq 0$ implies that $\mu(s)$ changes sign at $s=0$ so that the stability of the solution curve changes at the turning point.

Applying this to the constant density problem of section 2 we see that for small $s\in(-\delta,\delta)$ (and with $\alpha^2$ and $\gamma^2$ two known constants obtained from (\ref{mudot}) and (\ref{rhoddot})) 
$$\phi(s)=\phi_c+s\dot{\phi}(0)+O(s^2)=\phi_c+s\theta+O(s^2)$$
$$\rho(s)=\rho_c+s\dot{\rho}(0)+\frac{s^2}{2}\ddot{\rho}(0)=\rho_c-\gamma^2s^2+O(s^3)$$ which agrees with \cite{W} and section 2 above, and
$$\mu(s)=\mu(0)+s\dot{\mu}(0)+O(s^2)=\alpha^2s+O(s^2)$$ which tells us that the lower branch of solutions $s\in(-\delta,0)$ is stable and the upper branch $s\in(0,\delta)$ is unstable, as expected.

The solution branches for the conformal factor, lapse and shift vector found numerically in \cite{PY} for the XCTS formulation of the constraints are graphically similar to those studied above (LS methods were applied to this system in \cite{W} under the assumption that the system developed a  one dimensional kernel for sufficiently large initial data). It seems likely from this work that these branches also display an exchange of stability for a  broad class of initial data such that the non-degeneracy conditions (\ref{nondegen1}), (\ref{nondegen2}) are satisfied.

We conclude this section by mentioning a limitation on the application of the sub-super solution method that this work reveals. Solutions to elliptic equations obtained via the sub-super solution method are stable in the sense described in section 2 (see \cite{Sattinger}, \cite{Sattinger1}) so that the upper branch of solutions  studied here are unattainable by this method. The general CMC Lichnerowicz-York equation with an unscaled source on an asymptotically Euclidean manifold
\begin{equation}
\nabla^2\phi-r\phi+a\phi^{-7}+2\pi\rho\phi^5=0
\end{equation}
was studied in \cite{CBIY} using the sub-super solution method, where $r$ is proportional to the scalar curvature and $a$ the traceless part of the conformally transformed extrinsic curvature squared (see section XII of that work). They found an open set of values of $a$ and $\rho$ such that existence could be proven-uniqueness was not proven. 

If, as seems likely from this work and \cite{PY}, \cite{BOMP}, \cite{W},\cite{Holst} and \cite{HM}, that upper and lower branches of solutions exist for this equation for some combination of $a$ and $\rho$ then the sub-super solution and monotone iteration method will only converge to a stable solution. So if an upper branch of solutions is unstable, as above, then the sub-super solution method will not yield it. These comments are also relevant to the Einstein-Scalar field CMC Lichnerowicz-York equation as studied in \cite{CBIP} (see Theorem 8.8 in that work).

\section{Discussion}
In this work we have concentrated on the non-standard case of an unscaled fluid with no momentum. (For a  complete discussion of the role of conformal scaling of fluid sources see \cite{CB} or \cite{CBIY}). It is not of purely academic interest however because it serves as an excellent model of a  poorly scaled system such as the XCTS system and also warns of the dangers of not scaling the extrinsic curvature as was common in numerical evolutions which started from moment of time symmetry initial data, see section 5 in \cite{W}.

In studying (\ref{nonlinear}), we have been looking at the geometric problem of finding a  conformal factor that maps from a scalar flat metric to one with scalar curvature equal to $16\pi\rho$ (for more details see the classic paper \cite{KW} where the authors allow a  non-zero scalar curvature background metric). This is closely related to the Yamabe problem of finding a conformal factor that maps a given metric to one with constant scalar curvature. This has yielded many insights, particularly in relation to the existence theory of solutions to the Lichnerowicz-York equation, see \cite{BI} for an excellent review.

It is clearly important to understand the strengths and limitations of solution methods. With this in mind, future work should determine the stability of non-unique solutions found in other works, for example in \cite{PY}, \cite{Maxwell}.

The Lichnerowicz-York equation with an unscaled source also has a  variational formulation since (\ref{nonlinear}) is the Euler-Lagrange equation obtained by varying the following functional:

\begin{equation}
I[\phi]=\int\left( \frac{|\nabla\phi|^2}{2}-\frac{\pi\rho\phi^6}{3}\right)dv
\end{equation}

It would be interesting to check which solutions (stable/lower or unstable/upper or both) application of the Mountain Pass theorem to this toy model yields.

\section*{Acknowledgements}
I am grateful to Niall \'{O} Murchadha for helpful comments on an early draft of this work and an anonymous referee for suggestions that improved its presentation.





\section*{Appendix A}

We used MATLAB to perform the following integration for the Rayleigh quotient for the trial function $\eta_T$ given by (\ref{eta})

\begin{align}\label{ray}
\mu_1(\alpha)&\leq R(\alpha)\\ &=\frac{\int|\nabla\eta_T|^2-10\pi\rho\phi^4\eta_T^2dv}{\int\eta_T^2dv}\\
&=\frac{\int|\nabla\eta_T|^2-15\left(\frac{\alpha}{\alpha^2+r^2}\right)^2\eta_T^2dv}{\int\eta_T^2dv}\\
&=\frac{(3tan^{-1}(1/\alpha)\alpha^8-3\alpha^7+\alpha^5-6tan^{-1}(1/\alpha)\alpha^4-\alpha^3+3\alpha+3tan^{-1}(1/\alpha))}{     ((\alpha^4+2\alpha^2+1)\alpha^2(15tan^{-1}(1/\alpha)\alpha^4-15\alpha^3+6tan^-1(1/\alpha)\alpha^2-\alpha-tan^{-1}(1/\alpha)))}
\end{align}

Some specific values of the Rayleigh quotient are given in Table 1 below. 

\begin{table}[ht] \caption{ } \centering \begin{tabular}{c c c c} \hline\hline $\alpha$ & Rayleigh quotient  \\ [0.5ex]
 \hline 10,000&10.500 \\ 5&9.804  \\1&0\\0.99&-0.138\\0.75&-4.548\\0.5&-14.057\\0.25&-53.713\\0.01&-3.00x$10^4$\\ [1ex] \hline \end{tabular} \label{table:nonlin} \end{table}

\end{document}